\documentclass[preprint]{elsarticle}
\usepackage{amsbsy}
\usepackage{amsfonts}
\usepackage{amsmath}
\usepackage{amssymb}

\newdefinition{definition}{Definition}
\newtheorem{theorem}{Theorem}

\newproof{proof}{Proof}

\newtheorem{lemma}{Lemma}

\newdefinition{remark}{Remark}

\newdefinition{example}{Example}

\begin{document}
\begin{frontmatter}
\title{On the set of uniquely decodable codes with a given sequence of code word lengths}
\author[rvt]{Adam Woryna}
\ead{adam.woryna@polsl.pl}
\address[rvt]{Silesian University of Technology, Institute of Mathematics, ul. Kaszubska 23, 44-100 Gliwice, Poland}

\begin{abstract}
For every natural number $n\geq 2$ and every finite sequence $L$ of natural numbers, we consider the set $UD_n(L)$ of all uniquely decodable codes over an $n$-letter alphabet  with the sequence $L$ as the sequence of code word lengths, as well as its subsets $PR_n(L)$ and $FD_n(L)$ consisting of, respectively,  the prefix codes and the codes with  finite delay.  We derive the estimation for the quotient $|UD_n(L)|/|PR_n(L)|$, which allows to characterize those sequences $L$ for which the equality $PR_n(L)=UD_n(L)$ holds. We also characterize those sequences $L$ for which the equality $FD_n(L)=UD_n(L)$ holds.
\end{abstract}


\end{frontmatter}

\section{Preliminaries and the statement of the results}

Let $X$ be an alphabet with $n:=|X|\geq 2$ letters. We refer to a finite sequence
$$
C=(v_1,\ldots, v_m),\;\;\;m\geq 1
$$
of words over $X$  as {\it a code} and to the words  $v_i\in X^*$ ($1\leq i\leq m$) as the {\it code words}. In particular,   our convention differs  a bit from the more usual one, where  codes are considered as sets of words rather than  sequences of words. The code $C$ is called {\it uniquely decodable} if for all $l,l'\geq 1$ the equality $v_{i_1}v_{i_2}\ldots v_{i_l}=v_{j_1}v_{j_2}\ldots v_{j_{l'}}$
with $1\leq i_t, j_{t'}\leq m$ ($1\leq t\leq l$, $1\leq t'\leq l'$) implies $l=l'$ and $i_t=j_t$ for every $1\leq t\leq l$.  Thus every  uniquely decodable code must be an injective sequence of non-empty words. In the algebraic language,  one could say that the  code $C$ is uniquely decodable  if and only if the monoid  generated by the set $\{v_1, \ldots, v_m\}$ (with concatenation of words as the monoid operation) is a free monoid of rank $m$ freely generated by this set, or that this set is an $m$-element basis for this monoid.
If for all $1\leq i,j\leq m$ the condition: $v_i$ {\it is a prefix (initial segment) of} $v_j$ implies $i=j$, then $C$ is called a {\it prefix code}.

The prefix codes are the most useful examples of uniquely decodable codes and, in a sense, they are universal for all uniquely decodable codes. Namely, according to the Kraft-McMillan theorem (\cite{6}), for every finite sequence $L=(a_1, \ldots, a_m)$ of  natural numbers the following three statements are equivalent: (1) there exists a uniquely decodable code $C=(v_1, \ldots, v_m)$ with the sequence $L$ as the sequence of code word lengths, i.e. $|v_i|=a_i$ for every $1\leq i\leq m$; (2) there exists a prefix code $C'=(v_1', \ldots, v_m')$ with the sequence $L$ as the sequence of code word lengths; (3) the inequality  $\sum_{i=1}^m n^{-a_i}\leq 1$ holds.

Uniquely decodable codes of length $m\leq 2$ are exceptional, as every such a code has finite delay (\cite{2}). Recall that  a code $C$ has finite delay if  there is a number $t$ with the following property: picking up the consecutive letters of an arbitrary word $u\in X^*$ which can be   factorized into the code words, it is enough to pick up at most $t$ first letters of $u$ to be sure which code word begins $u$ (see also~\cite{0}). The smallest number $t$ with this property is called the {\it delay} of the code $C$. If such a number does not exist, then we say that the code has  {\it infinite delay}. Obviously, every   prefix code has   finite delay (which is not greater that the maximum length of a code word) and every  code  with  finite delay must be uniquely decodable. It turns out (see Section~6.1.2 in~\cite{4} and  Proposition~6.1.9 therein) that a code $C=(v_1,\ldots, v_m)$ has  infinite delay if and only if there is an infinite word $u\in X^\omega$ and two  factorizations
\begin{eqnarray*}
u&=&v_{i_1}v_{i_2}v_{i_3}\ldots,\\
u&=&v_{j_1}v_{j_2}v_{j_3}\ldots
\end{eqnarray*}
into code words such that $v_{i_1}\neq v_{j_1}$. If $m\geq 3$, then there are uniquely decodable codes of length $m$ which have  infinite delay.

\begin{example}\label{p1}
The code $C=(10,100,000)$ has  infinite delay because of the following two factorizations of the infinite word  $u=10^\infty$ into the code words:
\begin{eqnarray*}
10-000-000-000-\ldots,\\
100-000-000-000-\ldots.
\end{eqnarray*}
The code  $C$ is also uniquely decodable, as its {\it reverse}  $C^R=(01,001,000)$ is a  prefix code (we use the well known fact that a code is uniquely decodable if and only if its reverse  is uniquely decodable).
\end{example}

For every finite sequence $L$ of natural numbers we denote by $UD_n(L)$ the set of all  uniquely decodable codes over the alphabet $X$ with the sequence $L$ as the sequence of code word lengths. We also consider the subset $PR_n(L)\subseteq UD_n(L)$ of all prefix codes and the subset $FD_n(L)\subseteq UD_n(L)$ of all codes with finite delay. Thus, we have the inclusions $PR_n(L)\subseteq FD_n(L)\subseteq UD_n(L)$ and the set $UD_n(L)$ is non-empty if and only if the set $PR_n(L)$ is non-empty. If $L$ is constant, then each code in $UD_n(L)$ is a block code and we obviously have in this case: $PR_n(L)=UD_n(L)$. As we mentioned above, if the length of $L$ is 1 or  2, then  $FD_n(L)=UD_n(L)$.

The aim of this work is to characterize  those sequences $L$ for which the equality  $PR_n(L)=UD_n(L)$ holds, as well as those sequences $L$ for which $FD_n(L)=UD_n(L)$. For the first characterization, we modify the Kraft's procedure (\cite{5}) describing the construction of an arbitrary prefix code $C\in PR_n(L)$. This allows us to  obtain the following estimation for the quotient $|UD_n(L)|/|PR_n(L)|$ in the case when $L$ is non-constant.
\begin{theorem}\label{t4}
Let $L$ be a non-constant sequence such that the set $UD_n(L)$ is non-empty. Then we have
$$
\frac{|UD_n(L)|}{|PR_n(L)|}\geq 1+\frac{r_ar_b}{|PR_n((a,b))|},
$$
where  $a$ and $b$ are arbitrary two different values of $L$ and  $r_a$ (resp. $r_b$) is the number of those elements in $L$ which are equal to  $a$ (resp.  to $b$).
\end{theorem}

As a direct consequence of the above inequality, we obtain the following result.

\begin{theorem}\label{t2}
If the set  $UD_n(L)$ is non-empty, then the statements are equivalent:
\begin{itemize}
\item[(i)] $UD_n(L)=PR_n(L)$,
\item[(ii)]  $L$ is constant.
\end{itemize}
\end{theorem}

For the second characterization, we involve the Sardinas-Patterson algorithm (\cite{3}) and obtain the following  theorem.

\begin{theorem}\label{t3}
If the set  $UD_n(L)$ is non-empty, then the statements are equivalent:
\begin{itemize}
\item[(i)] $FD_n(L)=UD_n(L)$,
\item[(ii)] the length of   $L$ is not greater than  2 or, after  reordering the  elements of $L$, we have  $L=(a,a, \ldots, a, b)$, where   $a\mid b$.
\end{itemize}
\end{theorem}

\section{The Kraft's procedure for prefix codes}\label{r3}

Let $L$ be a finite sequence of natural numbers. We now present the Kraft's method for the construction of an arbitrary code $C\in PR_n(L)$ (\cite{5}), which can be used in deriving the formula for the number of elements in the set  $PR_n(L)$.

Let  $\widetilde{L}:=\{\nu_1, \nu_2, \ldots, \nu_l\}$ be the set of values of the sequence $L$ ordered from the smallest to the largest, i.e. $\nu_1<\nu_2<\ldots<\nu_l$ and let   $r_{\nu_i}$ ($1\leq i\leq l$) be the number of those elements in $L$ which are equal to  $\nu_i$.

To construct an arbitrary code $C\in PR_n(L)$ we proceed as follows. As the code words of  length $\nu_1$, we choose arbitrarily $r_{\nu_1}$ words among all the words of length $\nu_1$. This can be done in ${n^{\nu_1}\choose r_{\nu_1}}$ ways.  Next, we must arrange the chosen words in $r_{\nu_1}$ available positions of the sequence  $C$, which can be done in $r_{\nu_1}!$ ways. For the construction of the code words of length $\nu_2>\nu_1$, we  can use the remaining  $n^{\nu_1}-r_{\nu_1}$ available words of length  $\nu_1$ as  possible prefixes; for the final segments, we can take  arbitrary words of length $\nu_2-\nu_1$. Consequently, the number of ways to construct the code words of length $\nu_2$ is equal to
$$
{n^{\nu_2-\nu_1}\cdot(n^{\nu_1}-r_{\nu_1})\choose r_{\nu_2}}.
$$
Finally, as before, we arrange the chosen words in the sequence $C$, which can be done in $r_{\nu_2}!$ ways.

By continuing this reasoning, we see that  for every $1\leq i\leq l$  the code words of length $\nu_i$ can be chosen arbitrarily among the words of length $\nu_i$ which do not have as a prefix any previously chosen code word. If $N_i$ denotes the number of such available words, then we have
$$
N_1=n^{\nu_1},\;\;\;N_{i+1}=n^{\nu_{i+1}-\nu_i}(N_i-r_{\nu_i}),\;\;\;1\leq i<l.
$$
Hence, for each $1\leq i\leq l$ the code words of length  $\nu_i$  can be constructed and arranged in the sequence $C$  in  ${N_i\choose r_{\nu_i}}r_{\nu_i}!$ ways. Consequently, we obtain the following formula
for the cardinality of the set $PR_n(L)$:
\begin{equation}\label{f}
|PR_n(L)|=\prod_{i=1}^l {N_i\choose r_{\nu_i}}r_{\nu_i}!.
\end{equation}
In particular,   if  $PR_n(L)\neq \emptyset$, then  $n^{\nu_i}\geq N_i\geq r_{\nu_i}\geq 1$ for all  $1\leq i\leq l$, which implies: $N_i>r_{\nu_i}$ for $1\leq i<l$.

\begin{example}
Let $a,b\geq 1$ be natural numbers. If $C=(v,w)$ and $|v|=a$, $|w|=b$, then in the case  $a=b$ we have: $\widetilde{L}=\{a\}$ and $r_a=2$, and in the case $a\neq b$ we have: $\widetilde{L}=\{a,b\}$ and $r_a=r_b=1$. Hence, by formula (\ref{f}), we obtain:
$$
|PR_n((a,b))|=n^{a+b}-n^{\max(a,b)}.
$$
The last formula can also be derived directly by the  definition of a prefix code, that is without using (\ref{f}).
\end{example}

\section{The sets $UD_n(L)$ for particular sequences $L$}\label{r7}

The situation is much more complicated if we want to obtain the formula for the number of elements in the set  $UD_n(L)$. Nowadays, there are various algorithms testing the unique decodability of a code. We can use them and try to obtain the formula for $|UD_n(L)|$ in some particular cases of the sequence  $L$. In this section, we make the calculations for an exemplary sequence of length three, as well as for the sequences  from Theorem~\ref{t3}. Our calculations simultaneously provide the full characterization of the corresponding  sets $UD_n(L)$.

The calculations are based on the Sardinas-Patterson algorithm (\cite{3}), which claims that a code $C$ is uniquely decodable if and only if $C$ is an injective sequence of non-empty words and   $D_i\cap D_0=\emptyset$ for all  $i\geq 1$, where the sets  $D_i$ ($i\geq 0$) are defined recursively as follows: $D_0$ is the set of the code words, and  for  $i\geq 1$ the set  $D_i$ is the set  of all non-empty words $w\in X^*$ which satisfy  the following condition: $D_{i-1}w\cap D_0\neq \emptyset$ or $D_0w\cap D_{i-1}\neq\emptyset$, where $D_iw:=\{vw\colon v\in D_i\}$.

\subsection{The sequence $L=(2,3,3)$}

At first, let us assume that the unique code word of length two consists of two different letters. So, let  $(xy,w,v)$ be a code such that $x,y\in X$, $x\neq y$, and $w,v\in X^3$, where $w\neq v$. We have three possibilities: (1) $(w,v)=(xyz, txy)$ for some $z,t\in X$, (2) $(w,v)=(zxy, xyt)$ for some $z,t\in X$, (3) the word $xy$ is neither a prefix of $w$ nor a prefix of $v$ or it is neither a suffix (final segment) of $w$ nor a suffix of $v$. In the third case, we obviously have $(xy, w,v)\in UD_n(L)$. In the case (1), we have: $(xy,w,v)=(xy,xyz,txy)$. Now,  if $(z,t)=(x,y)$, then  $(xy)^3=wv$, and hence $(xy,w,v)\notin UD_n(L)$. So, let us assume that $(z,t)\neq (x,y)$. We have now four possibilities: $z\notin\{x,t\}$, or $z=x\neq t$, or $z=t\neq x$, or $z=x=t$. If $z\notin\{x,t\}$, then $D_1=\{z\}$, $D_2=\emptyset$ and hence $(xy, w,v)\in UD_n(L)$. If $z=x\neq t$, then $t\neq y$ and hence $D_1=\{x\}$,  $D_2=\{y,yx\}$, $D_3=\emptyset$, which implies  $(xy,w,v)\in UD_n(L)$. If $z=t\neq x$, then $D_1=\{t\}$, $D_2=\{xy\}$, which implies $(xy,w,v)\notin UD_n(L)$. If $z=t=x$, then $D_1=\{x\}$, $D_2=\{x,yx,xy\}$ and hence $(xy,w,v)\notin UD_n(L)$. Thus  in the case (1), we obtain: $(xy,w,v)\notin UD_n(L)$ if and only if $(w,v)=(xyx,yxy)$ or $(w,v)=(xyz,zxy)$ for some $z\in X$. Consequently, in this case, there are exactly  $n+1$ codes $(xy,w,v)$ which are non-uniquely decodable. In the case (2), by taking the reverse of a code $(xy,w,v)$ and  using the same reasoning,  we also obtain that there are exactly $n+1$ codes which are non-uniquely decodable. Hence, if $x\neq y$, then the number of elements in the  set
$$
C(x,y):=\{(xy,w,v)\colon w,v\in X^3\}\cap UD_n(L)
$$
is equal to
$$
|C(x,y)|=n^3(n^3-1)-2(n+1).
$$

We now calculate for a fixed $x\in X$ the number of elements in the set
$$
C(x):=\{(xx,w,v)\colon w,v\in X^3\}\cap UD_n(L).
$$
For any $w,v\in X^3\setminus \{xxx\}$ with $w\neq v$ there are two cases: (1) $xx$ is both the prefix of at least one of the words $w$, $v$ and the suffix of at least one of the words $w$, $v$, (2) $xx$ is neither a prefix of $w$ nor a prefix of $v$ or it is neither a suffix of $w$ nor a suffix of $v$. In the second case, we have $(xx,w,v)\in UD_n(L)$.  In the first case, we have two possibilities: (1a)  $(xx,w,v)=(xx, xxy, zxx)$ or (1b) $(xx,w,v)=(xx,yxx,xxz)$ for some $y,z\in X\setminus\{x\}$. Both in the case (1a) and in the case (1b) we have: if $y=z$, then for the code $(xx,w,v)$ we obtain: $xx\in D_2$ and hence $(xx,w,v)\notin UD_n(L)$. If $y\neq z$, then $D_1\subseteq \{y,z\}$ and $D_2=\emptyset$, and hence $(xx,w,v)\in UD_n(L)$. Thus the number of all codes of the form $(xx,w,v)\in C(x)$ satisfying  (1) is equal to $2((n-1)^2-(n-1))$, and the number of all codes  $(xx,w,v)\in C(x)$ satisfying  (2) is equal to  $(n^3-1)(n^3-2)-2(n-1)^2$. Hence
$$
|C(x)|=(n^3-1)(n^3-2)-2(n-1).
$$
Finally, we obtain
$$
|UD_n(L)|=\sum_{x,y\in X, x\neq y}|C(x,y)|+\sum_{x\in X} |C(x)|=n(n-1)(n^6+n^5-n^4-2n^2-6).
$$
For comparison, we obtain by the formula (\ref{f}):
$$
|PR_n(L)|=n(n-1)(n^6+n^5-n^4-2n^3-n^2).
$$

\subsection{The sequences of the form $L=(a,\ldots,a,b)$, where  $a\mid b$}

Let $L=(a,\ldots,a,b)$ be a sequence of length $m>1$ such that $a\mid b$. If  $a=b$, then $L$ is constant and hence $UD_n(L)=PR_n(L)$. Let us assume that $q:=b/a>1$. If $C\in UD_n(L)$, then obviously the code  $C$ must be of the form $(v_1, \ldots, v_{m-1}, w)$ for some pairwise different  words $v_i$ ($1\leq i\leq m-1$) of length  $a$ and the word $w$ of length $b$ which is not of the form $v_{j_1} v_{j_2}\ldots v_{j_q}$ for some  $j_\iota\in\{1,\ldots, m-1\}$, $\iota=1,2,\ldots, q$. Conversely, let us assume that $C$ is an arbitrary code of the form $(v_1, \ldots, v_{m-1}, w)$, where the words $v_i$, $w$ are as above. We show that $C\in UD_n(L)$. Indeed, since $|w|=qa$, we have $w=w_1\ldots w_q$ for some words $w_i$ ($1\leq i\leq q$) each of  length $a$. Let  $1\leq i_0\leq q$ be the smallest index such that $w_{i_0}\notin\{v_1, \ldots, v_{m-1}\}$. For $1\leq i<i_0$ let us consider the word $u_i:=w_{i+1}\ldots w_q$. Because of the minimality of  $i_0$, none of the words $u_i$ ($1\leq i<i_0$) is a prefix of $w$. Hence for $1\leq i<i_0$ we have $D_i=\{u_i\}$ and for $i\geq i_0$ we have $D_i=\emptyset$. Thus  $D_i\cap D_0=\emptyset$ for each $i\geq 1$, and hence $C\in UD_n(L)$. Now, by easy calculation, we obtain the following formula:
$$
|UD_n(L)|=n^a(n^a-1)\ldots(n^a-m+2)(n^b-(m-1)^{b/a}).
$$
For comparison, we have by (\ref{f}):
$$
|PR_n(L)|=n^a(n^a-1)\ldots (n^a-m+2)(n^b-(m-1)n^{b-a}).
$$
In particular,  the above formula for $|UD_n(L)|$ also works  in the case $a=b$.

\section{The proofs of the main results}\label{r5}

In this section we derive our main results.

{\renewcommand{\thetheorem}{\ref{t4}}
\begin{theorem}
Let $L$ be a non-constant sequence such that the set $UD_n(L)$ is non-empty. Then we have
$$
\frac{|UD_n(L)|}{|PR_n(L)|}\geq 1+\frac{r_ar_b}{|PR_n((a,b))|},
$$
where  $a$ and $b$ are arbitrary two different values of $L$ and  $r_a$ (resp. $r_b$) is the number of those elements in $L$ which are equal to  $a$ (resp.  to $b$).
\end{theorem}
\addtocounter{theorem}{-1}}
\begin{proof}
We will use the notations as in Section~\ref{r3}, i.e. by  $\widetilde{L}:=\{\nu_1, \ldots, \nu_l\}$, we denote the set of values of the sequence $L$ ordered from the smallest to the largest, i.e. $\nu_1<\nu_2<\ldots<\nu_l$ and by   $r_{\nu_i}$ ($1\leq i\leq l$) we denote  the number of those elements in $L$ which are equal to  $\nu_i$. Without losing generality, we can assume that $a<b$. Let $i_0, i_1\in \{1,\ldots, l\}$ be indices corresponding to the values $a,b\in\widetilde{L}$, i.e. $\nu_{i_0}=a$, $\nu_{i_1}=b$. Let us fix two different letters $0, 1\in X$ and let $PR_{n, a, b}(L)$ be the subset of $PR_n(L)$ consisting of prefix codes with  the words  $w_a:=0^{a-1}1$,  $w_b:=0^{b-1}1$ as code words.

An arbitrary code $C\in PR_{n, a, b}(L)$ can be constructed as follows. At first, for every $1\leq i<i_0$, we choose the code words of length $\nu_i$ and  arrange them in the sequence $C$ in the same way as in the Kraft' procedure keeping only in mind not to choose the ``zero'' word $0^{\nu_i}$. Thus for every $1\leq i< i_0$  the number of available words for the code words of length $\nu_i$ is equal to $N_i-1$ and hence, the number of ways to construct these code words and arrange them in the sequence $C$ is equal to
$$
{N_i-1\choose r_{\nu_i}}r_{\nu_i}!=\frac{N_i-r_{\nu_i}}{N_i}\cdot {N_i\choose r_{\nu_i}}r_{\nu_i}!=n^{\nu_i-\nu_{i+1}}\frac{N_{i+1}}{N_i}\cdot {N_i\choose r_{\nu_i}}r_{\nu_i}!.
$$
Note that  for $1\leq i<i_0$ we have $N_i>r_{\nu_i}$, and hence the above number is indeed positive.

For the construction of the code words of length $\nu_{i_0}=a$, we  also remember that $0^a$ can not be a code word. Beside of that, the word $w_a=0^{a-1}1$ must be  a code word. Hence, we need to choose  $r_a-1$ words of length $a$ among all   $N_{i_0}-2$ available words, and next to arrange the chosen words together with the word $w_a$ in $r_a$ available positions in the sequence $C$. Thus the number of ways to construct the code words of length $a$ and arrange them in $C$ is equal to
$$
{N_{i_0}-2\choose r_a-1}r_a!=\frac{r_a}{N_{i_0}-1}\cdot n^{\nu_{i_0}-\nu_{i_0+1}}\cdot\frac{N_{i_0+1}}{N_{i_0}}\cdot {N_{i_0}\choose r_a}r_a!.
$$
Since $i_0<l$, we have $N_{i_0}>r_a$ and hence this number is indeed positive.

In the next step, we construct for every  $i_0<i<i_1$ the code words of length  $\nu_i$. We are still restricted  to the words different from  $0^i$ and hence, the number of ways to  do this is equal to
$$
{N_i-1\choose r_{\nu_i}}r_{\nu_i}!=n^{\nu_i-\nu_{i+1}}\frac{N_{i+1}}{N_i}\cdot {N_i\choose r_{\nu_i}}r_{\nu_i}!.
$$

For the construction of the code words of length $\nu_{i_1}=b$, we must remember that the word $w_b=0^{b-1}1$ is a code word. But now, we can choose the ``zero'' word $0^b$ as a code word. Hence, we need to choose   $r_b-1$ words  among $N_{i_1}-1$ available words. In consequence, the number of ways to construct the code words of length $b$ and arrange them in the sequence $C$ is equal to
$$
{N_{i_1}-1\choose r_b-1}r_b!=\frac{r_b}{N_{i_1}}\cdot {N_{i_1}\choose r_b}r_b!.
$$
Since $N_{i_1}\geq r_b$,  this number  is indeed positive.

In the final step, we construct for every  $i_1<i\leq l$  the code words of length $\nu_i$. This construction can be done in ${N_i\choose r_{\nu_i}}r_{\nu_i}!$ ways, as we can follow exactly in the same way as in the Kraft's procedure.

As a result of the above procedure, we see that the number of ways to construct an arbitrary code from the set $PR_{n,a,b}(L)$ is equal to
$$
\prod_{1\leq i<i_1}\left(n^{\nu_i-\nu_{i+1}}\cdot \frac{N_{i+1}}{N_i}\right)\cdot \frac{r_a}{N_{i_0}-1}\cdot \frac{r_b}{N_{i_1}}\cdot \prod_{i=1}^l {N_i\choose r_{\nu_i}}r_{\nu_i}!=\frac{r_ar_b}{n^b(N_{i_0}-1)}|PR_n(L)|.
$$
Hence, we obtain
$$
|PR_{n,a,b}(L)|=\frac{r_ar_b}{n^b(N_{i_0}-1)}|PR_n(L)|.
$$
By the inequality  $N_{i_0}\leq n^{\nu_{i_0}}=n^a$, we have:
$$
|PR_{n,a,b}(L)|\geq \frac{r_ar_b}{n^b(n^a-1)}|PR_n(L)|=r_ar_b\frac{|PR_n(L)|}{|PR_n((a,b))|}.
$$

To finish the proof, it suffices to observe that if  $C\in PR_{n,a,b}(L)$, then for the reverse $C^R$ we have $C^R\in UD_n(L)\setminus PR_n(L)$. Indeed, the words $(w_a)^R$, $(w_b)^R$ are code words in $C^R$ and the word $(w_a)^R=1(0^{a-1})$ is a prefix of the word $(w_b)^R=1(0^{b-1})$.  Since for the arbitrary codes $C_1$, $C_2$ we have $C_1=C_2 \Leftrightarrow C_1^R=C_2^R$, we conclude the equality
$$
|PR_{n,a,b}(L)|=|\{C^R\colon C\in PR_{n,a,b}(L)\}|.
$$
In consequence, we obtain
$$
|UD_n(L)|\geq|PR_n(L)|+|PR_{n,a,b}(L)|\geq |PR_n(L)|\left(1+\frac{r_ar_b}{|PR_n((a,b))|}\right).
$$
\qed
\end{proof}

{\renewcommand{\thetheorem}{\ref{t3}}
\begin{theorem}
If the set  $UD_n(L)$ is non-empty, then the statements are equivalent:
\begin{itemize}
\item[(i)] $FD_n(L)=UD_n(L)$,
\item[(ii)] the length of   $L$ is not greater than  2 or, after  reordering the  elements of $L$, we have  $L=(a,a, \ldots, a, b)$, where   $a\mid b$.
\end{itemize}
\end{theorem}
\addtocounter{theorem}{-1}}
\begin{proof}
At first, we show the implication  (ii)$\Rightarrow$ (i). If  $L$ has the length at most 2, then according to~\cite{2}, every code in   $UD_n(L)$ has finite delay. If $L=(a,a, \ldots, a, b)$, where $a\mid b$, then we have two possibilities: $a=b$ or $a\neq b$. In the first case $L$ is constant and then each  $C\in UD_n(L)$ is a prefix code, which implies that $C$ has  finite delay.

If  $a\neq b$, then each $C\in UD_n(L)$ also has  finite delay. To show this, let us assume that we have picked up the first $b$ letters of a word  $u\in X^*$, for which we only know that it is  factorizable into code-words. Let $w$ be the prefix of $u$ of length $b$. We have two possibilities:  $w$ is not a code word or  $w$ is a code word.

In the first case, since there are no code words longer than $b$ and all the code words shorter than $b$ have the length  $a$, the prefix of length  $a$ in the word $u$ must be a code word and $u$ begins with  this code word.

In the second case, $w$ is the code word with which the word $u$ starts. To show this, let us suppose contrary that $u$ does not begin with  $w$. Since $w$ is the only code word of length $b$ and all the other code words have the length $a<b$, there must be the maximum number  $k\geq 1$ such that the word $w_1\ldots w_k$ is a prefix of $w$, where each  $w_i$ is a code word of length $a$. Now, if  $w_1\ldots w_k=w$, then $C$ would not be uniquely decodable. So, let us assume that  $w_1\ldots w_k$ is a proper prefix of $w$. In particular, we obtain  $ka<b$. Hence there is a code word $v$ such that $w_1\ldots w_kv$ is a prefix of $u$. Now, if $|v|=a$, then in view of the inequality $ka<b$ and the divisibility $a\mid b$, we would have $(k+1)a\leq b$ and consequently, the word $w_1\ldots w_kv$ would be a prefix of  $w$, contrary to the maximality of $k$. Hence $|v|=b$, which implies $v=w$. Thus $w$ must be a prefix of  $w_1\ldots w_kw$. But then the divisibility  $a\mid b$ implies the equality  $w=(w_1\ldots w_k)^sw_1\ldots w_{k'}$ for some $s\geq 0$, $1\leq k'<k$, and again we have a contradiction with the assumption that $C$ is uniquely decodable. Thus in each case it is enough to pick up at most $b$ letters of the word $u$ to know  which code word begins this word.

To show (i)$\Rightarrow$ (ii) let us assume that $L$ does not satisfy the condition (ii).  We must show that there is a code $C\in UD_n(L)$ with  infinite delay. The sequence  $L$ has the length at least three and $L$ is not constant. Let  $a$ and $b$ be the two smallest values of $L$ and let us assume that $a<b$. Let us define in the same way as in the proof of Theorem~\ref{t4} the set  $\widetilde{L}=\{\nu_1, \ldots, \nu_l\}$ of the values of $L$,  the sequence $(r_{\nu_i})_{1\leq i\leq l}$,  the words $w_a=0^{a-1}1$, $w_b=0^{b-1}1$, and the subset  $PR_{n, a, b}(L)\subseteq PR_n(L)$. In particular, we have: $\nu_1=a$, $\nu_2=b$.

If $r_{b}>1$, then we can use the construction of the code described in the proof of Theorem ~\ref{t4} and obtain a code $C\in PR_{n,a,b}(L)$ such that one of its code words of length $b$ is $0^{b}$. Since  $C^R\in UD_n(L)$ and the words $1(0^{a-1})$, $1(0^{b-1})$, $0^{b}$ are code words in  $C^R$, we see, by analogy to Example~\ref{p1}, that $C^R$ has  infinite delay.

If $r_{b}=1$ and  $L$ has at least three values, then there is the smallest   $1\leq i_0\leq l$ such that  $\nu_{i_0}>b$. Similarly as in the previous case, we can use the construction from the proof of Theorem~\ref{t4} and construct a code $C\in PR_{n,a,b}(L)$, such that   $0^{\nu_{i_0}}$ is one of the code words. Then the words  $1(0^{a-1})$, $1(0^{b-1})$ and $0^{\nu_{i_0}}$ are code words in $C^R\in UD_n(L)$, and similarly as above, we obtain that $C^R$ has  infinite delay.

The last case is  when  $r_{b}=1$ and the only values of $L$ are  $a$ and  $b$. Since $L$ does not satisfy the condition (ii), we obtain $a\nmid b$. Let  $\eta\in\{1,\ldots, a-1\}$ be the remainder from the division of $b-a$ by $a$. Then we have $b-a=qa+\eta$ for some integer $q\geq 0$. Since  $L$ has the length at least three, we have $2\leq r_{a}<n^a$. Thus, there is an injective code  $C$ with the sequence  $L$ as the sequence of code word lengths and such  that  the words $1^{a}0^{b-a}$, $1^{a}$, $0^{a}$ are  the code words and the word   $1^{a-\eta}0^\eta$ is not a code word. Then the infinite word $1^{a}0^\infty$ has two factorizations into code words:
\begin{eqnarray*}
1^{a}-0^{a}-0^{a}-\ldots,\\
1^{a}0^{b-a}-0^{a}-0^{a}-\ldots.
\end{eqnarray*}
Thus it is enough to show that  $C$ is uniquely decodable. For this aim, we need to show that $D_i\cap D_0=\emptyset$ for all  $i\geq 1$, where the sets  $D_i$ ($i\geq 0$) are constructed according to the  Sardinas-Patterson algorithm, i.e. $D_0$ is the set of the code words, and  for  $i\geq 1$ the set  $D_i$ consists of all non-empty words $w\in X^*$ for which  the following condition holds: $D_{i-1}w\cap D_0\neq \emptyset$ or $D_0w\cap D_{i-1}\neq\emptyset$.

Let $S$ be the set of all non-empty words which are proper suffixes (final segments) of the code words.  Obviously, every word in  $S$ is shorter than $b$. Hence the intersection $S\cap D_0$ contains only the code words of length $a$ which are the proper suffixes of the other code words. Since $1^{a}0^{b-a}$ is the only code word of length greater than $a$, the set  $S\cap D_0$ consists of the code words of length  $a$ which are suffixes of the code word $1^{a}0^{b-a}$.
Thus, if $b-a>a$, then $S\cap D_0=\{0^{a}\}$. If $b-a<a$, then  $\eta=b-a$, and hence $S\cap D_0=\emptyset$, as the suffix of length $a$ in the code word $1^{a}0^{b-a}$ is   $1^{2a-b}0^{b-a}=1^{a-\eta}0^\eta$, which, by our assumption, is not a code word. Hence, we obtain:
\begin{equation}\label{alt}
S\cap D_0=\left\{
\begin{array}{ll}
\emptyset, &\mbox{\rm if}\;b-a<a,\\
\{0^{a}\}, &\mbox{\rm if}\;b-a>a.
\end{array}
\right.
\end{equation}

\begin{lemma}\label{l2}
For every  $i\geq 1$ the inclusion $D_i\subseteq S$ holds.
\end{lemma}
\begin{proof}[of Lemma~\ref{l2}]
By the definition of $D_1$, we have $D_1=\{0^{b-a}\}\subseteq S$. Let us assume inductively that $D_i\subseteq S$ for some $i\geq 1$. Let $w\in D_{i+1}$ be an arbitrary word. Then  $D_{i}w\cap D_0\neq \emptyset$ or $D_0w\cap D_{i}\neq\emptyset$. In the first case, we have $vw\in D_0$ for some nonempty word $v\in X^*$, i.e.  $w$ is a proper suffix of the code word $vw$, and hence $w\in S$. In the second case, we have $vw\in D_i$ for some  $v\in X^*$. By the inductive assumption, we obtain $vw\in S$, i.e. $vw$ is a proper suffix of a code word, and hence $w$ is also a proper suffix of this code word. Thus $w\in S$ and consequently, we have $D_{i+1}\subseteq S$.\qed
\end{proof}

Suppose now that $D_i\cap D_0\neq \emptyset$ for some $i\geq 1$.  Since $D_i\subseteq S$, we obtain by (\ref{alt}): $D_i\cap D_0=\{0^{a}\}$. Consequently, there is the smallest number  $i\geq 1$ such that  $0^{\lambda a}\in D_i$ for some integer $\lambda\geq 1$. Since $D_1=\{0^{b-a}\}$ and $a\nmid b$, we have $i\geq 2$. By the definition of the set $D_i$ we have $D_{i-1}0^{\lambda a}\cap D_0\neq\emptyset$ or $D_00^{\lambda a}\cap D_{i-1}\neq\emptyset$. In the first case, we obtain that  $v0^{\lambda a}$ is a code word for some  $v\in D_{i-1}$. Since $|v0^{\lambda a}|>a$, it must be $v0^{\lambda a}=1^{a}0^{b-a}$, and hence $v=1^{a}0^{b-(\lambda+1)a}$. Since $v\in S$ and $|v|>a$, the word $v$ must be a suffix of the code word $1^{a}0^{b-a}$ and we obtain a contradiction because the word  $1^{a}0^{b-(\lambda+1)a}$ is not a suffix of $1^{a}0^{b-a}$.

In the second case, we have $v0^{\lambda a}\in D_{i-1}$ for some code word $v\in D_0$. Since $D_{i-1}\subseteq S$ and $|v0^{\lambda a}|>a$, the word $v0^{\lambda a}$ must be a suffix of the code word $1^{a}0^{b-a}$. Since $v$ is a code word, we obtain $|v|=a$. Thus  $v$ must be of the form $0^a$ or $1^a$ or $1^{a-\gamma}0^{\gamma}$ for some integer $0<\gamma\leq a-1$. If $v=0^{a}$, then  $0^{(\lambda+1)a}=v0^{\lambda a}\in D_{i-1}$ and we obtain a contradiction with the minimality of $i$. If $v=1^{a}$, then the word $v0^{\lambda a}=1^{a}0^{\lambda a}$ is a suffix of the code word $1^{a}0^{b-a}$; consequently,  it must be $\lambda a=b-a$, and again we  have a contradiction with $a\nmid b$. Hence, it must be   $v=1^{a-\gamma}0^{\gamma}$ for some integer $0<\gamma\leq a-1$. But then  $v0^{\lambda a}=1^{a-\gamma}0^{\gamma+\lambda a}$. Consequently, the word
$1^{a-\gamma}0^{\gamma+\lambda a}$ is a suffix of the code word $1^{a}0^{b-a}$. In particular, we obtain $\gamma+\lambda a=b-a$. But, since $b-a=qa+\eta$ and $0<\eta\leq a-1$, we obtain $\gamma=\eta$ and $\lambda=q$. Thus $v=1^{a-\eta}0^\eta$ and we have a contradiction with the assumption that  $1^{a-\eta}0^\eta$ is not a code word. Consequently $D_i\cap D_0=\emptyset $ for every $i\geq 1$. Thus $C\in UD_n(L)$, which completes  the proof of Theorem~\ref{t3}.\qed
\end{proof}

\end{document}